\documentclass{article}
\usepackage{hyperref}
\usepackage{amscd,amsmath,amsthm,amssymb}
\usepackage{enumerate,varioref}
\usepackage{epsfig}
\usepackage{graphicx}
\usepackage{mathtools}
\usepackage{tikz}
\usepackage{xcolor,colortbl}
\newtheorem{thm}{Theorem}

\theoremstyle{definition}

\newtheorem{question}{Question:}
\theoremstyle{remark}

\newtheorem{rem}[thm]{Remark}
\numberwithin{equation}{section}

\newcommand{\half}{\frac{1}{2}}

\title{The Limiting Eigenvalue Distribution of Iterated $k$-Regular Graph Cylinders}
\author{Clark Alexander\footnote{Nousot Applied Research Group}, Tara Nenninger\footnote{NARG and University of Vermont}, Danielle Tucker\footnote{NARG and University of Illinois at Chicago} \\ email \href{mailto:clark@nousot.com}{clark@nousot.com}}

\begin{document}

\maketitle

\begin{abstract}
We explore the limiting empirical eigenvalue distributions arising from matrices of the form \[A_{n+1} = \begin{bmatrix}
A_n & I\\ I & A_n
\end{bmatrix} , \]where $A_0$ is the adjacency matrix of a $k$-regular graph.  We find that for bipartite graphs, the distributions are centered symmetric binomial distributions, and for non-bipartite graphs, the distributions are asymmetric. This research grew out of our work on neural networks in $k$-regular graphs.

Our original question was whether or not the graph cylinder construction would produce an expander graph that is a suitable candidate for the neural networks being developed at Nousot.  This question is answered in the negative for our computational purposes. However, the limiting distribution is still of theoretical interest to us; thus, we present our results here.
\end{abstract}

\section{Introduction}

\hspace{4mm} Our study began by looking for expander graphs that satisfy criteria by which we could write neural networks with a graph theoretic structure yielding a large spectral gap.  However, we find that the spectral gap of the normalized matrices corresponding to graph cylinders converges to zero, namely
\[
|\lambda_1 - \lambda_2| \in O(\varepsilon) \,,
\] where $\lambda_1$ and $\lambda_2$ are the first and second largest eigenvalues of the normalized adjacency matrix, respectively. Nonetheless, these graphs exhibit interesting spectral properties.   

In this paper, we will consider normalized matrices.  If $G$ is a $k$-regular graph, then we consider the matrix
\[
\frac{1}{k} A_G \,, 
\]
where $A_G$ is the corresponding adjacency matrix.\\ 

 The paper is set up as follows.  In Section 2, we define graph cylinders explicitly and label the vertices in such a way that the corresponding adjacency matrix has a convenient form for iteration.  We then make mention of how this relates to the Markov transition matrix for a classical random walk on a graph.  Section 3 gives the main theorems of this paper and presents a proof which shows that normalized iterated graph cylinders limit to Dirac $\delta$ distributions centered at zero,  and that non-normalized graph cylinders have asymmetric distributions.  Section 4 works through the first example of hypercubes, where we find that the spectral gap for $N$ iterations is exactly $\frac{2}{N+2}$.  These results are significantly more intuitive than those of a general graph, as hypercubes are a well-studied family of graphs.  We conclude with a discussion of this work and future work we intend to pursue.  Further, we discuss a few implications of this work in machine learning and artificial intelligence.

\section{Graph Cylinders}

\hspace{5mm}From classical geometry, we consider the cylinder of a set $S$ as the set of points $S\times[0,1]$.  For a graph $G = (V,E)$, the cylinder of $G$ is defined as $G_c = (V_c,E_c)$. The vertex set, $V_c$, is two copies of the vertex set of $G$, i.e.
\[
V_c = V\times \{0,1\}\,.
\]
The edge set, $E_c$, is two copies of the original set $E$ and the corresponding edges between copies of vertices.  Explicitly,
\[
E_c = \{(i_0,j_0),(i_1,j_1),(i_0,i_1) | i_k \in V\times\{k\} \}\,.
\]
\textbf{Example 1:}

Consider the graph $C_3$, which is isomorphic to an equilateral triangle:
\[
\begin{tikzpicture}
\node[circle,draw](1L) at (.7,1){1};
\node[circle,draw](2L) at (0,0){2};
\node[circle,draw](3L) at (1.4,0){3};
\foreach \from/\to in {1L/2L,1L/3L,2L/3L}
\draw (\from)--(\to);
\end{tikzpicture}
\]The graph $C_3$ has adjacency matrix
\[
A = \begin{bmatrix}
0 & 1 & 1\\
1 & 0 & 1\\
1 & 1 & 0
\end{bmatrix}.
\]
\pagebreak 

Now, consider the graph cylinder of $C_3$:

\[
\begin{tikzpicture}
\node[circle,draw](1L) at (.7,1){1L};
\node[circle,draw](2L) at (0,0){2L};
\node[circle,draw](3L) at (1.4,0){3L};
\node[circle,draw](1R) at (5.7,1){1R};
\node[circle,draw](2R) at (5,0){2R};
\node[circle,draw](3R) at (6.4,0){3R};
\foreach \from/\to in {1L/2L,1L/3L,2L/3L,1R/2R,1R/3R,2R/3R}
\draw (\from)--(\to);
\draw (1L)..controls(3.2,2).. (1R);
\draw (2L)..controls(2.5,1)..(2R);
\draw (3L)..controls(3.6,-1)..(3R);
\end{tikzpicture}
\]Since we connect each vertex in the first copy of the graph to its corresponding vertex in the second copy of the graph, we have the resulting adjacency matrix:

\[
A_c = \begin{bmatrix}
0 & 1 & 1 & 1 & 0 & 0\\
1 & 0 & 1 & 0 & 1 & 0\\
1 & 1 & 0 & 0 & 0 & 1\\
1 & 0 & 0 & 0 & 1 & 1\\
0 & 1 & 0 & 1 & 0 & 1\\
0 & 0 & 1 & 1 & 1 & 0
\end{bmatrix} = \begin{bmatrix}
A & I \\I & A
\end{bmatrix}\,
\]

\hspace{-5mm}\textbf{Example 2:}

Consider the graph cylinder constructed from a $C_6$: 

\[
\begin{tikzpicture}
\node[circle,draw](1T) at (1.1,5){1T};
\node[circle,draw](2T) at (-1,4){2T};
\node[circle,draw](3T) at (.5,3){3T};
\node[circle,draw](4T) at (2,3){4T};
\node[circle,draw](5T) at (4,4){5T};
\node[circle,draw](6T) at (2.6,5){6T};
\node[circle,draw](1B) at (1.1,0){1B};
\node[circle,draw](2B) at (-1,-1){2B};
\node[circle,draw](3B) at (.5,-2){3B};
\node[circle,draw](4B) at (2,-2){4B};
\node[circle,draw](5B) at (4,-1){5B};
\node[circle,draw](6B) at (2.6,0){6B};
\foreach \from/\to in {1T/2T,2T/3T,3T/4T,4T/5T,5T/6T,6T/1T,1B/2B,2B/3B,3B/4B,4B/5B,5B/6B,6B/1B,1T/1B,2T/2B,3T/3B,4T/4B,5T/5B,6T/6B}
\draw (\from)--(\to);
\end{tikzpicture}
\]
\[\text{The adjacency matrix of a $C_6$ graph is }
A = \begin{bmatrix}
0 & 1 & 0 & 0 & 0 & 1\\
1 & 0 & 1 & 0 & 0 & 0\\
0 & 1 & 0 & 1 & 0 & 0\\
0 & 0 & 1 & 0 & 1 & 0\\
0 & 1 & 0 & 1 & 0 & 1\\
1 & 0 & 0 & 0 & 1 & 0\\
\end{bmatrix},\] 

\[\text{ and the resulting adjacency matrix of its graph cylinder is }
A_c = \begin{bmatrix}
A & I \\
I & A
\end{bmatrix}. 
\] 
We have the following immediate properties of graph cylinders:
\begin{itemize}
\item	If $G$ is a $k$-regular graph, then the graph cylinder, $G_c$, is a ($k+1$)-regular graph.
\item If $A_G$ is the adjacency matrix of $G$, then $\begin{bmatrix}
A_G & I \\
I & A_G
\end{bmatrix}$ is the adjacency matrix of $G_c$.
\item Let $g_0$ be the girth of the original graph, $G$. The girth of the graph cylinder is $\min(g_0,4)$.  Suppose the girth of the upper graph is 5 or more. The resulting graph cylinder contains a 4-cycle. That is, if $(i,j)$ is an edge in $G_0$, $E_c$ contains $(i_0,j_0),(j_0,j_1),(j_1,i_1),(i_1,i_0)$.
\end{itemize}

The normalized matrix for a $k$-regular graph is $\frac{1}{k}A$.  Consider the iterates
\begin{eqnarray}
A_0 & = & A_G\\
\label{iterate}A_{n+1} & = & \begin{bmatrix}
A_n & I \\
I & A_n
\end{bmatrix}\\
W_n & = & \frac{1}{k+n} A_n
\end{eqnarray}

$W_n$ denotes that $W$ is the matrix corresponding to a random walk on graph $G$ iterated $n$ times.  Since $G$ is $k$ regular, the degree matrix, $D$, is simply $k\cdot I$. The walking matrix is
\[
W = D^{-1}A = \frac{1}{k}A,
\] 
which we have called $W_0$.

\section{The Distribution of Eigenvalues}

\hspace{4mm} Given a $k$-regular graph and its graph cylinder, iteration as shown in equation \ref{iterate} produces a family of regular graph cylinders.  We now attempt to answer the following:
\begin{question}
	What is the empirical distribution of eigenvalues of $\lim_{n\rightarrow \infty} W_n$?

\end{question}

In answering this, we shall first look at the eigenvalues of $A_n$. Let $\{k = \lambda_1 > \lambda_2 \ge \dots, \ge \lambda_n\}$ be the eigenvalues of $A_0$. 

\begin{thm}
	The eigenvalues of $A_1$ are the eigenvalues of $A_0$ shifted by one.  Without regarding explicit sorting of eigenvalues, the set is
	\[
	\{\lambda_1 +1, \lambda_1 - 1,\lambda_2+1,\lambda_2-1,\dots,\lambda_n+1,\lambda_n-1\}.
	\]
\end{thm}

\begin{proof}
	Consider the matrix $A_1 = \begin{bmatrix}
	A_0 & I \\ I & A_0
	\end{bmatrix}$.
	The eigenvalues of $A_1$ are given by the solutions to the equation\[
	\det (A_1 - \lambda\cdot I) = 0. 
	\]
	In block form, we have
\begin{equation}
A_1 - \lambda\cdot I = \begin{bmatrix}
A_0 - \lambda\cdot I & I\\
I & A_0 - \lambda\cdot I
\end{bmatrix}.
\end{equation}
	
From basic linear algebra and specifically Lemma 6.3 in \cite{S}, we know that the determinant of a block matrix with square blocks is given by
\begin{equation}
\det \left(\begin{bmatrix}
M_{11} & M_{12} \\
M_{21} & M_{22}
\end{bmatrix}\right) = \det(M_{11} - M_{12}M_{22}^{-1}M_{21})\det(M_{22})
\end{equation}
so long as $M_{22}$ is invertible. In the special case where $M_{11} = M_{22}$ and $M_{12} = M_{21}$,	we have
\begin{equation}
\det\left(\begin{bmatrix}
M_{11} & M_{12} \\ M_{12} & M_{11}
\end{bmatrix}\right) = \det(M_{11} + M_{12})\det(M_{11}-M_{12}).
\end{equation}

In our case, where $M_{11} = A-\lambda\cdot I$ and $M_{12} = I$, we have
\begin{equation}
\det\left(\begin{bmatrix}
A_0 - \lambda\cdot I & I\\
I & A_0 - \lambda\cdot I
\end{bmatrix}\right) = \det(A_0 - (\lambda-1)\cdot I)\det(A_0 - (\lambda+1)\cdot I) = 0 \label{plus1}.
\end{equation}

We see that the left product satisfies the eigenvalue equation of $A_0$, where each eigenvalue has a value one greater than before.  The second product satisfies the eigenvalue equation of $A_0$, but with each eigenvalue shifting the value down by one.

\end{proof}

\begin{thm}
	For any $k$-regular graph $G$, the limiting distribution of the normalized graph cylinder is the Dirac $\delta$ distribution centered at zero. Moreover, the limiting distribution of the normalized Laplacian matrix of the graph cylinder is the uniform distribution on $[0,2]$.
\end{thm}

\begin{proof}
	Consider the set of eigenvalues $\Lambda_0$ of $A_G$ given by
	\[
	\Lambda_0 = \{\lambda_1,\cdots,\lambda_n\}, 
	\]
	where the multiplicity of the $j^{th}$ eigenvalue is $m_j$. 
	
	Our graphs are simple graphs, so the trace $tr(A_G) = 0$, and the weighted sum of eigenvalues is also zero, namely
	\[
	\sum_{j=1}^{n} m_j \lambda_j = 0.
	\]

	Now, let $M_n$ represent the multiplicities of the set $\Lambda_n$.  Thus, our initial set is
	\[
	(\Lambda_0,M_0).
	\]
	
	Following equation \ref{plus1}, our next set is
	\[
	(\Lambda_1,M_1) = (\Lambda_0 - 1,M_0) \cup (\Lambda_0 + 1,M_0).
	\]

	Two steps later, we have
	\[
	(\Lambda_{2},M_{2}) = (\Lambda_0-2,M_0)\cup (\Lambda_0,2M_0) \cup (\Lambda_0+2,M_0).
	\]

	Continuing inductively, after $2n$ steps we have
	\begin{equation}
	(\Lambda_{2n},M_{2n}) = \bigcup_{k=-n}^{n} \left(\Lambda_0 + 2k,\binom{n}{|k|}M_0\right).
	\end{equation}
	
	Therefore, each eigenvalue in this set follows a shifted symmetric binomial distribution centered at its corresponding element in $\Lambda_0$. We denote a binomial distribution centered at $\lambda$ with probability $p$ and $n$ steps by 
	\[
	\mathbb{B}_{\lambda}(n,p). 
	\]
    
    \begin{rem}
		If $X \sim B(n,p)$, a standard binomial distribution, then our new notation suggests that $X - np + \lambda \sim \mathbb{B}_{\lambda}(n,p)$.  
	\end{rem} 
	Because $p= \half$ in our case, without loss of generality, we shall simply denote
	\[
	\mathbb{B}_{\lambda}(n,1/2) =\mathbb{B}_{\lambda}(n).
	\]Let $X$ be a random variable that follows the distribution of eigenvalues which we are studying. Then we have that  after $N$ steps, 
	\begin{equation}
	X \sim \frac{1}{|V|}\sum_{j=1}^{k} m_j \mathbb{B}_{\lambda_j}(N),
	\end{equation} where $|V|$ is the number of vertices in graph $G$. \\ 
	\begin{rem}
		We use $|V|$ since $\sum_j \frac{m_j}{|V|} = 1$.
	\end{rem} 
\vspace{4mm}
	This takes care of the case where the largest eigenvalue is $k+n$, the regularity of the graph plus the number of iterations.  Now, we want to consider normalized eigenvalues, i.e. the eigenvalues of  $W_n$. Let's denote the distribution of these normalized eigenvalues by $\hat{\Lambda}_n$. Let $\hat{X}\sim \hat{\Lambda}_n$. Then after N steps,
	\begin{equation}
	\hat{X} \sim \frac{1}{k+N} X \sim \frac{1}{(k+N)|V|} \sum_{j=1}^{k} m_j \mathbb{B}_{\hat{\lambda}_j}(N), 
	\end{equation}
	where $\hat{\lambda}_j = \frac{\lambda_j}{k+N}$.\\

	Now, take the limit as $N\rightarrow \infty$.
	\begin{eqnarray}
	\lim_{N\rightarrow \infty} \hat{\Lambda}_N = \lim_{N\rightarrow \infty} \frac{1}{(k+N)|V|}\sum \mathbb{B}_{\hat{\lambda}_j} (N) = \lim_{N\rightarrow \infty}\frac{\mathbb{B}_0(N)}{(k+N)}
	\end{eqnarray}
	
	As $N\rightarrow \infty$, the probability of any individual eigenvalue occurring converges to zero in probability. We then consider the cumulative distribution function:
	\[
	Pr(\hat{X}\le a) = \sum_{k=0}^{\lfloor aN\rfloor} \frac{\binom{N}{k}}{2^N} - \mu
	\]
	
	There is no simple closed form solution to this, but in \cite{Mc} we see that this is approximated closely by a truncated error function with slope $N$, or more explicitly, some monotonically increasing function of $N$. We can see that this holds by considering the ratio of binomial coefficients close to one half.  As $\binom{n}{k} = \binom{n}{n-k}$, we know this is symmetric about $\binom{n}{n/2}$, and that $\binom{n}{n/2}$ is the maximum value.  So we consider for some large $n$ and some $k>0$
	\begin{equation}
	\frac{\binom{2n}{n-k}}{\binom{2n}{n}} = \frac{n(n-1)\cdots (n-k_1)}{(n+k)(n+k-1)\cdots(n+1)} \approx \left(\frac{n}{n+k}\right)^k \approx \left(1 + \frac{k}{n}\right)^{-k}. 
	\end{equation}
	
	If $k$ is large, this can be further approximated to $e^{-k}$. Thus, the decay away from the center is exponential when $n$ is large.  
    Asymptotically, the cdf becomes a logistic with slope $n$, i.e.
	\begin{eqnarray}
	\sum_{j=0}^{k} \frac{\binom{n}{j}}{2^n} - \mu \approx \left(1 + \exp\left(\frac{-nk}{2}\right)\right)^{-1}. 
	\end{eqnarray}

In the limit, this becomes a step function at zero.  Thus, taking the derivative of our cumulative distribution function, we find that our probability distribution function is that of the Dirac $\delta$ distribution centered at zero.  \\

In the case of the normalized Laplacian, we know that for $k$-regular graphs
\[
\mathcal{L} = I - \frac{1}{k}A.
\]

Thus, the eigenvalues of $\mathcal{L}$ are 
\[
\Lambda_{\mathcal{L}} = \{ 1 - \lambda/k  \hspace{1mm}|\hspace{1mm}  \lambda \in \Lambda_{A}\}.
\]
	That is, simply transform the set by $\frac{1-x}{k}$.  In the limit of a graph cylinder, the probability of any single eigenvalue of the normalized adjacency matrix becomes zero, with max $[-1,1]$. This transforms our Laplacian eigenvalues into a uniform distribution on $[0,2]$.	\\

\end{proof}
\begin{figure}[ht!]
\centering
\includegraphics[width=5in]{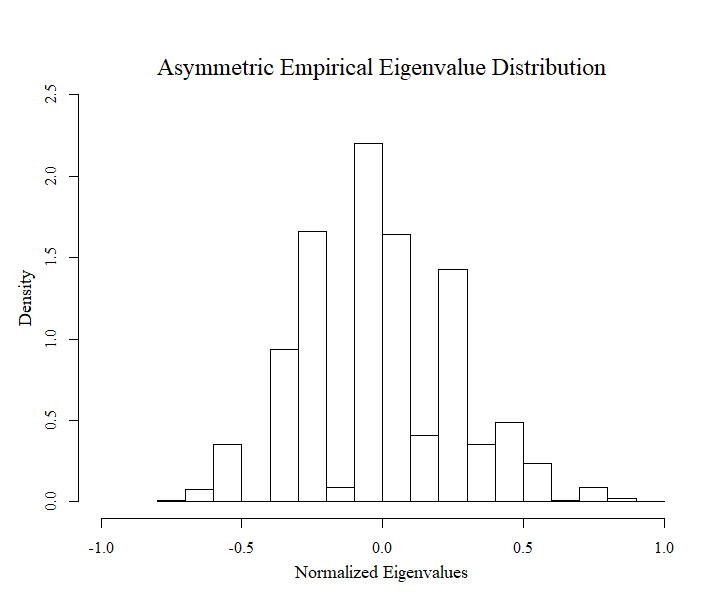}
\caption{This histogram depicts the empirical distribution of eigenvalues after 10 iterations on a non-bipartite graph, the K5 graph. Note the asymmetry. Recall that this is the sum of multiple shifted symmetric binomial distributions.}
\label{hypercube}
\end{figure}

\section{Example: Hypercubes}

The hypercube graphs $Q_n$ are defined geometrically as graphs with $2^n$ points, each of which has degree $n$.  
The $2^n$ points are those in $\{0,1\}^n$, and the edges are drawn such that two vertices are endpoints of an edge if they differ in exactly one bit. That is, the square is isomorphic to $Q_2$, and the cube is isomorphic to $Q_3$.  The adjacency matrix for $Q_2$ is
\[
A = \begin{bmatrix}
0 & 1 & 0 & 1\\
1 & 0 & 1 & 0\\
0 & 1 & 0 & 1\\
1 & 0 & 1 & 0\\
\end{bmatrix}
\].

We can compute the eigenvalues without much algebra.  Since this is a regular graph, the largest eigenvalue is the degree of regularity, $2$.  Since this graph is bipartite, the smallest eigenvalue is $-2$.  We see one degeneracy by looking at the first and third columns and a second degeneracy in the second and fourth columns. Thus, the list of eigenvalues is
\[
\Lambda_0 = \{-2,0,0,2\}
\]

We scale this to fit within $[-1,1]$ by dividing by the degree of regularity and we see
\[
\hat{\Lambda}_0 = \{-1,0,0,1\}
\]

Now, taking the graph cylinder of $Q_2$, we have the eigenvalue set
\[
\Lambda_1 = \{-3,-1,-1,1,-1,1,1,3 \} = \{-3,-1,-1,-1,1,1,1,3\}.
\]

When normalized, this becomes 
\[
\hat{\Lambda}_1 = \{-1,-1/3,-1/3,-1/3,1/3,1/3,1/3,1\}.
\]

Repeating this, we see $\Lambda_n$ is the centered binomial distribution with $2^n$ values.  That is, the multiplicity of $\pm k$ is $\binom{n}{|k|}$.

Thus, our normalized distribution $W_n$ is given by the following mass function:
\[
P\left(\lambda = 1 - \frac{2k}{n}\right) = \frac{\binom{n}{n-k}}{2^n }= \frac{\binom{n}{k}}{2^n}
\]

When we take the limit as $n$ goes to infinity and letting $\alpha = k/n$ we get:
\begin{eqnarray}
P(\lambda = 1-2\alpha) & = & \lim_{n\rightarrow \infty} \frac{\binom{n}{\alpha n}}{2^n} \nonumber \\
& = & \lim_{n\rightarrow \infty} \frac{1}{2^n}\sqrt{\left(\frac{1}{2\pi\alpha(1-\alpha )n }\right)}\frac{1}{\alpha^{\alpha}(1-\alpha)^{1-\alpha}}
\end{eqnarray}

Here we have used Stirling approximation for the binomial coefficient.  However, given any particular $\alpha$, as $n$ goes to infinity we get identically zero probabilities.  The solution is to consider a continuous probability.  Let's consider the cumulative distribution:
\[
P(\lambda \le 1-2\alpha) =  \frac{1}{2^n}\sum_{j=0}^{k} \binom{n}{k} 
\]

We know that this approximates a logistic function with slope $f(n)$, where $f$ is some monotonically increasing function \cite{Mc} at $\alpha = \half$.  Considering the limit
\begin{equation}
P(\lambda \le 1-2\alpha) =\left\{ \begin{array}{cc} 
	0 & \alpha< \half\\
	1 & \alpha >\half
	\end{array}  \right\},
\end{equation} this gives us a step function at zero.

\begin{figure}[ht!]
\centering
\includegraphics[width=5in]{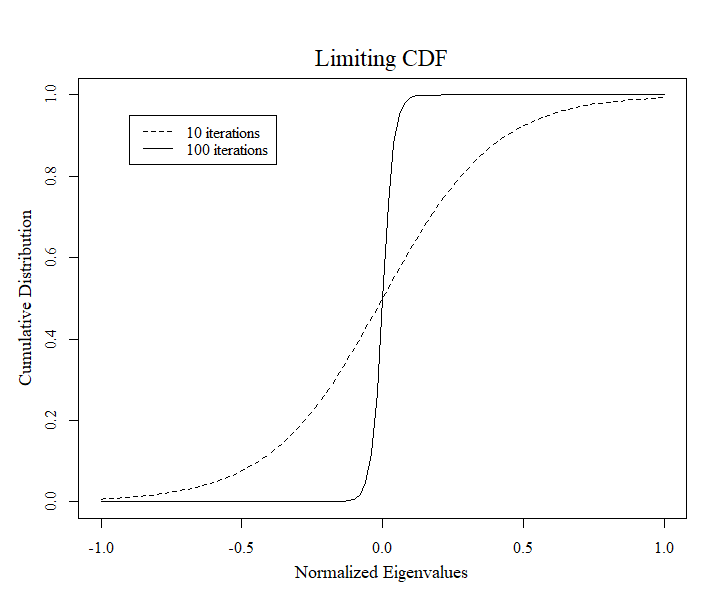}
\caption{This plot depicts the limiting nature of the CDF of eigenvalues, using the hypercube as an example. Note that as the number of iterations increases, the CDF approaches a step function centered at 0.}
\label{asymmetry}
\end{figure}

Thus, because the probability distribution function is the derivative of the cumulative distribution, we have a Dirac $\delta$-distribution.

\begin{equation}
P(\lambda = 1-2\alpha) = \left\{ 
\begin{array}{cc}
0 &\alpha \ne \half \\
\infty & \alpha = \half.\end{array}\right\}
\end{equation}

\section{Conclusion and Future Work}

Our original idea was to produce a family of good expander graphs which arose from a new technique.  The graphs of Lubotzky, Phillips, and Sarnak \cite{LPS}, as well as the ``zig-zag" product of Reingold, Vadhan, and Widgerson \cite{RVW} are two well-known constructions of families of expander graphs. Oher related works of Margulis, Gabber, and Galil provide types of expander graphs with which are ``concentrators" \cite{Mar,GG}.  Our work turns out to have none of these properties.  We began with the simple notion that a graph cylinder, as we define it, doubles the number of vertices while increasing the degree of regularity by exactly one.  These graphs do not, however, increase the spectral expansion constant.  Rather, they reduce it as much as possible.  This is particularly well seen in the spectrum of the normalized Laplacian matrix, which limits to a uniform distribution.  Thus, the smallest positive eigenvalue is within $\varepsilon$ of zero.  However, because the adjacency matrices of graph cylinders do become sparse very quickly, there may be some use in applying these graphs as models for neural networks.  Even though the spectral expansion is diminished, the diameter of a graph cylinder grows linearly while the number of vertices grows exponentially. Therefore, using a graph cylinder as the architecture of a neural network naturally produces network highways, which is still of interest to us. In future work, we also wish to explore the role of deformed Laplacians as models for neural networks, taking advice from \cite{FS,GHN,M}.  From the authors' viewpoint, there is also interest in exploring the role of deformations of graphs in quantum walks on graphs, as well as using deformed Laplacians to study the spectral expansion of directed graphs.

\end{document}